\documentclass[10pt,journal,doublecolumn]{IEEEtran}
\usepackage{amsmath,amssymb,threeparttable,placeins}
\usepackage{lipsum}
\usepackage{bm}
\usepackage{multicol}
\usepackage{array}

\usepackage[mathscr]{euscript}
\usepackage[dvips]{graphicx}
\usepackage[usenames]{color}
\usepackage{amsfonts}
\usepackage{latexsym}
\usepackage{subfigure}
\usepackage[format=hang]{subfig}
\usepackage[justification=centering]{caption}
\usepackage{floatrow}
\usepackage{multirow}
\usepackage{graphicx}
\usepackage{epstopdf}
\usepackage[english]{babel}
\DeclareMathAlphabet{\mathpzc}{OT1}{pzc}{m}{it}

\newtheorem{theorem}{{{{\textit{Theorem}}}}}

\newtheorem{lemma}{{{{\textit{Lemma}}}}}
\newtheorem{corollary}{{{{{\textit{Corollary}}}}}}

\newtheorem{remark}{{{{\textit{Remark}}}}}

\begin{document}
	\title{Generalized Tan-Arlery-Rabaste-Lehmann-Ovarlez Lower Bound on Ambiguity Function of a Set of Sequences With Mismatched Filters}
	\author{Shibsankar Das,~\IEEEmembership{Member,~IEEE}
		\thanks{Shibsankar Das is with the Department of Mathematics, University of Patliputra, Patna, India  (E-mail: shibsankar@ieee.org).
		}
	}
	\maketitle
	\begin{abstract}
		In this paper, a lower bound on the maximum ambiguity function (AF) sidelobes of a set of unimodular sequences is formulated for the desired low-ambiguity-zone (LAZ). Our main idea is to introduce a set of mismatched filters associated to a set of unimodular sequences and two weight vectors for the delay and Doppler shifts, respectively. The length of mismatched filter maybe different to the length of unimodular sequence. The proposed lower bound on the maximum AF sidelobes for the desired LAZ can be treated as an extension of Tan-Arlery-Rabaste-Lehmann-Ovarlez lower bound, published in 2020, which dealt with the conventional correlation of sequences. 
	\end{abstract}
	
	\section{Introduction} 
	\label{sec:intro} 
	\IEEEPARstart{T}{HE} next-generation wireless technologies are expected to feature significant delay and Doppler selectivity \cite{2023WeijieChinaCommunications},  \cite{2026YiViterboStandardsMagazine}. In this scenario, ambiguity function (AF) plays a pivotal role in determining the performance of the system. For example, a near-ideal thumbtack-shaped AF can resolve the Doppler sensitivity within a local AF region in high-mobility integrated sensing and communication (ISAC) systems \cite{2024KechengIEEEJrSelAreSensor}, \cite{2026ChengjuanIEEEWComLett}. Also, low integrated sidelobe level (ISL) of AF can improve the detection performance of orthogonal time-frequency space (OTFS)-ISAC systems \cite{2024IEEETVTKechengFan}, \cite{2026FengLiuTsinghuaScienceTechnology}. Therefore, 
	AF shaping can provide a comprehensive characterization of unimodular sequences used in the delay–Doppler domain. To control the transmission power and to avoid the non-linearity in the system, the transmitted sequences are considered to be unimodular (i.e., each sequence element has unity modulus). It is thus essential to design a set of unimodular sequences with low auto-and cross-AF sidelobe level for high-mobility ISAC systems \cite{2024KechengIEEEJrSelAreSensor}-\nocite{2026ChengjuanIEEEWComLett}\nocite{2024IEEETVTKechengFan}\cite{2026FengLiuTsinghuaScienceTechnology} and radar systems \cite{2024ChenIEEETSP},  \cite{2026XiangqingHuIEEETAES}. The AF includes the classical correlation function as a special case when the Doppler shift is zero. There have been many correlation lower bounds in the literature \cite{1974Welch}\nocite{1999LevenshteinIEEETIT}\nocite{2014LiuIEEETIT}-\cite{2019ArleryIEEETIT}. Most recently, a lower bound for the maximum correlation sidelobes has been investigated using mismatched filter in  \cite{2020Tan-Arlery-Rabaste-Lehmann-OvarlezTIT}.
	
	The maximum AF magnitude directly influences the detection capabilities of unimodular sequences for the multiple moving targets under the strong clutter \cite{2022FulaiIEEETGRS}. It is expected to have ideal AF shape on the entire global AF region \cite{2019YangIEEETAES}. That is, AF has zero sidelobes on the entire delay-Doppler plane except for the origin. However, there does not exist any unimodular sequence with zero AF sidelobes on entire AF region due to the lower bound on AF volume \cite{2020NajafabadiIEEETAES}. To overcome this limitation, one feasible scheme is to design a unimodular receiving filter associated to the transmitted unimodular sequence such that the desired AF shape is achieved within a local AF region \cite{2024KechengIEEEJrSelAreSensor}-\nocite{2026ChengjuanIEEEWComLett}\nocite{2024IEEETVTKechengFan}\nocite{2026FengLiuTsinghuaScienceTechnology}\nocite{2024ChenIEEETSP}\cite{2026XiangqingHuIEEETAES}. The local AF region with low AF sidelobes is called a low ambiguity zone (LAZ) in delay-Doppler domain \cite{2022ZhifanJSAComLiu}. Compared to the receiving matched filters, the low AF sidelobes within the desired LAZ can be achieved by designing a set of receiving mismatched filters associated to the transmitting set of unimodular sequences with some loss in processing gain (LPG) \cite{2022FulaiIEEETGRS}. In recent years, to obtain the desired LAZ, mismatched filter design associated to a unimodular sequence has gained tremendous research interest  \cite{2024KechengIEEEJrSelAreSensor}-\nocite{2026ChengjuanIEEEWComLett}\nocite{2024IEEETVTKechengFan}\nocite{2026FengLiuTsinghuaScienceTechnology}\nocite{2024ChenIEEETSP}\cite{2026XiangqingHuIEEETAES}. Thus, in this paper, it is not a major concern to design both the sets of unimodular sequences and mismatched filters. In addition, there have been extensive research works to design unimodular sequences with LAZ in the literature \cite{2025LiyingIEEETIT}\nocite{2026BingshengIEEETCOM}\nocite{2026MeiyueIEEETSP}\nocite{2026XiupingIEEESPL}-\cite{2026ZhengIEEECOML} and some works also investigate theoretical bounds on the maximum AF magnitude with LAZ \cite{2022ZhifanJSAComLiu}, \cite{2025BingshengIEEETIT}\nocite{2025LingshengIEEETIT}-\cite{2026GangsanIEEETIT}. Driven by the above background, it remains largely open: What is the lower bound on the maximum AF magnitude that can be achieved using a mismatched filter associated to a local AF region? Therefore, it is interesting to investigate a lower bound on the maximum AF magnitude using a mismatched filter associated to the LAZ. 
	
	In this paper, inspired by \cite{2020Tan-Arlery-Rabaste-Lehmann-OvarlezTIT} and \cite{2025LingshengIEEETIT}, we present a lower bound of the maximum AF magnitude (periodic or aperiodic AF) associated to the LAZ using mismatched filter. Specifically, we consider a set  $\{\textit{\textbf{x}}^m\}_{m=1}^{M}$ of $M$ unimodular sequences of length $L_x$ and a set $\{\textit{\textbf{y}}^n\}_{n=1}^{M}$ of $M$ associated mismatched filters of length $L_y$. The AF between the sequence $\textit{\textbf{x}}^m$ and associated filter $\textit{\textbf{y}}^n$ can be referred to auto-AF (AAF) when $m=n$ and the AF between the sequence $\textit{\textbf{x}}^m$ and filter $\textit{\textbf{y}}^n$ can be referred to cross-AF (CAF) when $m\neq n$. A lower bound on the maximum AF sidelobe levels of these AF values is derived for the desired LAZ by considering two weight vectors for the delay and Doppler shifts. We show that the proposed lower bound includes the existing lower bounds from \cite{1974Welch} and \cite{2020Tan-Arlery-Rabaste-Lehmann-OvarlezTIT} as the special cases in zero-Doppler cut.
	
	The rest of the paper is organized as follows. In Section {\ref{Sec:Preli}}, we present some definitions and preliminary ideas. Then, we formulate a lower bound on the maximum AF sidelobe levels for the desired LAZ with mismatched filters in Section \ref{Sec:Proposed:Lower:Bound}. In Section \ref{Sec:Rela:Existing:Litr}, we compare the new bound with some previous bounds. Finally, we conclude our new results in Section \ref{Sec:conclusion}. 
	
	\section{Background and Definitions}
	\label{Sec:Preli}
	\subsection{Periodic and Aperiodic Cross-Ambiguity Function}
	Let $\{\textit{\textbf{x}}^m\}_{m=1}^{M}$ be a set of $M$ unimodular sequences of length $L$. Each sequence of $\{\textit{\textbf{x}}^m\}_{m=1}^{M}$ with 	$\textit{\textbf{x}}^m=(x_0^m,x_1^m, \cdots, x_{L-1}^m)$ and $|x_l^m|=\frac{1}{L}$ for $l=0,1,\cdots, L-1$ and $m=1,2,\cdots,M$. Thus, the energy of the $m$-th unimodular sequence $\textit{\textbf{x}}^m$ is given by $\sum_{l=0}^{L-1}|x_l^m|^2=1$ for each $m\in \{1,2,\cdots, M\}$. The ambiguity function (AF) between two sequences $\textit{\textbf{x}}^m$ and $\textit{\textbf{x}}^n$ with $m,n \in \{1,2,\cdots, M\}$ at delay shift $\tau$ and Doppler shift $\nu$ is defined by
	\begin{align}
		{\bm \phi}^{m,n}(\tau,\nu)=\sum_{l=0}^{L-1}x_{l-\tau}^{m}\left(x_l^{n}\right)^*e^{\frac{j2\pi \nu l }{L}},
	\end{align}where $0\leq |\tau|\leq L-1$ and $0\leq |\nu|\leq L-1$. When $m=n$, ${\bm \phi}^{m,n}(\tau,\nu)$ is called auto-ambiguity function (AAF), denoted by ${\bm \phi}^{m}(\tau,\nu)$; otherwise, it is called cross-ambiguity function (CAF). Note that ${\bm \phi}^{m,n}(\tau,\nu)$ becomes the classical correlation function when the Doppler shift $\nu=0$, denoted by ${\bm \theta}^{m,n}(\tau)$ \cite{2010HaoIEEESPL}, \cite{2017Shibsankar} at the delay shift $\tau$. The AF can be defined into two types, periodic and aperiodic, based on the properties of each unimodular sequence $\textit{\textbf{x}}^m$ in the set $\{\textit{\textbf{x}}^m\}_{m=1}^{M}$. More specifically, for any $l <0$ or $ l \geq L$: $x_l^m=\begin{cases}
			x_{(l)_L}^m \quad \text{in the periodic AF},\\
			0 \quad \quad \  \text{in the aperiodic AF},
		\end{cases}$ where $(l)_L$ is calculated modulo $L$. 
	
	\subsection{The Maximum AAF Magnitude and Maximum CAF Magnitude Associated with the LAZ}
	For $1\leq Z_x \leq L$ and $1\leq Z_y \leq L$, we define the low ambiguity zone (LAZ) ${\bm \prod}$ in the delay-Doppler domain as ${\bm \prod}=\Big\{ (\tau, \nu): -Z_x<\tau <Z_x \ \text{and} \ -Z_y<\nu <Z_y \Big\}$. In practical scenarios, $Z_x$ and $Z_y$ are determined by the maximum delay shift and the maximum Doppler frequency shift. The maximum AAF magnitude ${\bm \phi}_{\max}^\text{AAF}$ and maximum CAF magnitude ${\bm \phi}_{\max}^\text{CAF}$ associated with the LAZ ${\bm \prod}$ are defined by  
	\begin{align}
		\left({\bm \phi}_{\max}^\text{AAF}\right)^2= \max_{(\tau,\nu)\neq (0,0) \atop |\tau|<Z_x, |\nu|<Z_y}\Big\{ {\big|{\bm \phi}^{m}(\tau,\nu)\big|^2} \Big\}, \\
		\left({\bm \phi}_{\max}^\text{CAF}\right)^2= \max_{|\tau|<Z_x, |\nu|<Z_y \atop m\neq m'}\Big\{ {\big|{\bm \phi}^{m,m'}(\tau,\nu)\big|^2} \Big\}. 
	\end{align}Therefore, the maximum AF magnitude ${\bm \phi}_{\max}$ associated to the LAZ ${\bm \prod}$ is defined by ${\bm \phi}_{\max}^2= \max \Big\{ \left({\bm \phi}_{\max}^\text{AAF}\right)^2, \left({\bm \phi}_{\max}^\text{CAF}\right)^2 \Big\}$.
	
	\section{The Generalized Mismatched Filter Bound}
	\label{Sec:Proposed:Lower:Bound}

	Let $\{\textit{\textbf{x}}^m\}_{m=1}^{M}$ be a set of $M$ unimodular sequences of length $L_x$. Let $\{\textit{\textbf{y}}^m\}_{m=1}^{M}$ be a set of $M$ associated mismatched filters of length $L_y$. Let us assume that the length of each mismatched filter is $L_y=L_x+2L_s$, where $L_s \in \mathbb{N}$. For practical applications, the inner product between a unimodular sequence $\textit{\textbf{x}}^m$ and its associated mismatched filter $\textit{\textbf{y}}^m$ in delay-Doppler domain should be maximized and with another mismatched filter $\textit{\textbf{y}}^n$ should be minimized for each $1\leq m,n \leq M$  with $m\neq n$.
	 
	\subsection{The AF Between a Set of Unimodular Sequences With Its Associated Mismatched Filter}
	
	The periodic and aperiodic CAF between the unimodular sequence $\textit{\textbf{x}}^m$ of length $L_x$ and its associated mismatched filter $\textit{\textbf{y}}^n$ of length $L_y$ is defined by 
	\begin{align}
		\label{AF:Mismatched:filter}
		{\bm \phi}^{m,n}(\tau,\nu)=\sum_{l=0}^{L_y-1}x_{l-\tau-L_s}^{m}\left(y_l^{n}\right)^*e^{\frac{j2\pi \nu l }{L_y}},
	\end{align}where $0\leq |\tau|< L_x+L_s$ and $0\leq |\nu|\leq L_y-1$. When $\nu=0$, it is the correlation function between $\textit{\textbf{x}}^m$ and $\textit{\textbf{y}}^n$, denoted by ${\bm \theta}^{m,n}(\tau)$. We now show that phase shifting of each mismatched filter $\textit{\textbf{y}}^m$ does not effect any AF sidelobe level for each Doppler frequency shift $\nu$ and delay shift $\tau=0$.
	
	\begin{lemma}[Phase Shifting]
		When the delay shift $\tau=0$ and Doppler frequency shift $\nu$ with $0\leq |\nu|< L_y$, the AF ${\bm \phi}_1^{m,m}(0,\nu)$ between $\textit{\textbf{x}}^m$ and  mismatched filter $\textit{\textbf{y}}_1^m$ can be expressed by
		 ${\bm \phi}_1^{m,m}(0,\nu)=\big|{\bm \phi}^{m,m}(0,\nu)\big|$, where $\textit{\textbf{y}}_1^m=\textit{\textbf{y}}^me^{j \psi_m}$ for some $\psi_m \in [-\pi, \pi]$  and ${\bm \phi}^{m,m}(0,\nu)$ is the AF between $\textit{\textbf{x}}^m$ and $\textit{\textbf{y}}^m$.
	\end{lemma}
	
	\begin{IEEEproof}
		From the definition (\ref{AF:Mismatched:filter}), for each $\nu$, the AF ${\bm \phi}^{m,m}(0,\nu)$ between $\textit{\textbf{x}}^m$ and $\textit{\textbf{y}}^m$ is given by
		\begin{align}
			\label{t=0:v:AF}
			{\bm \phi}^{m,m}(0,\nu)&=\sum_{l=0}^{L_y-1}x^m_{l-0-L_s}(y_l^m)^* e^{\frac{j2\pi \nu l }{L_y}} \nonumber \\
			&=\big|{\bm \phi}^{m,m}(0,\nu)\big|e^{j\psi_m}, 
		\end{align}where $\psi_m$ is the phase of ${\bm \phi}^{m,m}(0,\nu)$ for the Doppler frequency  $\nu$. We define a new phase shifted mismatched filter $\textit{\textbf{y}}_1^m$ as follows $ \textit{\textbf{y}}_1^m=\textit{\textbf{y}}^me^{j\psi_m}$ for some $\psi_m \in [-\pi, \pi]$. We now calculate the AF between $\textit{\textbf{x}}^m$ and $\textit{\textbf{y}}_1^m$. Using (\ref{t=0:v:AF}), we have 
		\begin{align}
			\label{AF:t:0:Doppler:shifted:MMF}
		&{\bm \phi}_1^{m,m}(0,\nu)
		= \left(\sum_{l=0}^{L_y-1}x^m_{l-L_s}(y_{l}^m)^*  e^{\frac{j2\pi \nu l }{L_y}}\right)e^{-j\psi_m} \nonumber \\
		&= \left(\big|{\bm \phi}^{m,m}(0,\nu)\big|e^{j\psi_m }\right)e^{-j\psi_m} =\big|{\bm \phi}^{m,m}(0,\nu)\big|.
		\end{align}This completes the proof.
	\end{IEEEproof}
	\begin{remark}
	According to (\ref{AF:t:0:Doppler:shifted:MMF}), we can say that $\textit{\textbf{y}}^m$ and $\textit{\textbf{y}}_1^m$ have same AF sidelobes when $\psi_m$'s are different for each Doppler shift $\nu$ with $0\leq |\nu|\leq L_y-1$.
	\end{remark}
	The maximum AF magnitude ${\bm \phi}_{{\max}_0}^{m,m}$ between  $\textit{\textbf{x}}^m$ and $\textit{\textbf{y}}^m$ associated with the LAZ ${\bm \prod}$ when the delay shift $\tau=0$ is defined by
	$\Big|{\bm \phi}_{{\max}_0}^{m,m}\Big|^2=\max_{0\leq |\nu|\leq Z_y-1}\left\{\big|{\bm \phi}^{m,m}(0,\nu)\big|^2\right\}$.
	
\subsection{The Generalized Mismatched Filter Bound Associated with the LAZ}
We first consider the LAZ ${\bm \prod}$ with $1\leq Z_x\leq L_x$ and $1\leq Z_y \leq L_x\leq L_y$. For the unimodular sequence set $\{\textit{\textbf{x}}^m\}_{m=1}^{M}$, we define $Z_y$ Doppler shifted unimodular sequences for each $m$ as $
		\widehat{\textit{\textbf{x}}}^{m,u}=\left(x^m_0,x^m_1e^{\frac{j2\pi u}{L_y}}, \cdots, x^m_{L_x-1}e^{\frac{j2\pi(L_x-1)u}{L_y}}\right)$, where $m=1,2,\cdots, M$ and $u=0,1,\cdots, Z_y-1$. We define sequence $\textit{\textbf{x}}^{m,u}$ of length $L_x+L_y-1$ as 
	\begin{align}
		\label{Doppler:Unimodular:Seq}
		\textit{\textbf{x}}^{m,u}=\left(\widehat{\textit{\textbf{x}}}^{m,u}\cdot e^{\frac{j2\pi L_su}{L_y}}, \textit{\textbf{0}}_{1\times (L_y-1)} \right),
	\end{align}where $\textit{\textbf{0}}_{1\times (L_y-1)}$ is the zero vector of length $L_y-1$. The circulant matrix $\textbf{circ}\left(\textit{\textbf{x}}^{m,u}\right)$ of each  sequence $\textit{\textbf{x}}^{m,u}$ has size $(L_x+L_y-1)\times (L_x+L_y-1)$. For the mismatched filter set $\{\textit{\textbf{y}}^n\}_{n=1}^{M}$, we define $Z_y$ Doppler shifted mismatched filters of length $L_x+L_y-1$ for each $n$ as follows:
\begin{align}
	\label{Doppler:MMF:Filter}
	&\textit{\textbf{y}}^{n,v}=\left(y^n_{L_s}e^{\frac{j2\pi L_sv}{L_y}},y^n_{L_s+1}e^{\frac{j2\pi (L_s+1)v}{L_y}}, \cdots, y^n_{L_y-1}e^{\frac{j2\pi (L_y-1)v}{L_y}},\right. \nonumber \\
	& \quad \left. 0, 0, \cdots , 0, y^n_{0}, y^n_{1}e^{\frac{j2\pi v}{L_y}}, \cdots, y^n_{L_s-1}e^{\frac{j2\pi (L_s-1)v}{L_y}}\right),
\end{align}where $n=1,2,\cdots, M$ and $v=0,1,\cdots, Z_y-1$. The circulant matrix $\textbf{circ}\left(\textit{\textbf{y}}^{n,v}\right)$ of each  sequence $\textit{\textbf{y}}^{n,v}$ has size $(L_x+L_y-1)\times (L_x+L_y-1)$. According to (\ref{Doppler:Unimodular:Seq}) and (\ref{Doppler:MMF:Filter}), we can express AF values as follows:

{\small 
\begin{align}
	&\Big(\textbf{circ}\left(\textit{\textbf{x}}^{m,u}\right)\Big) \cdot \Big(\textbf{circ}\left(\textit{\textbf{y}}^{n,v}\right) \Big)^H \nonumber \\
	&= \textbf{circ}\left(\Big[ {\bm \phi}^{m,n}(0,u-v), {\bm \phi}^{m,n}(-1,u-v)e^{\frac{j2\pi u}{L_y}}, \cdots, \right. \nonumber \\
	&  {\bm \phi}^{m,n}(-L_{\tau},u-v)e^{\frac{j2\pi L_{\tau}u}{L_y}},{\bm \phi}^{m,n}(L_{\tau},u-v)e^{\frac{-j2\pi L_{\tau}u}{L_y}},  \nonumber \\
	& \qquad \qquad \qquad \qquad \left. \cdots, {\bm \phi}^{m,n}(1,u-v)e^{\frac{-j2\pi u}{L_y}} \Big] \right), 
\end{align}}where $L_{\tau}=L_x+L_y-1$. We consider two weight vectors ${\bm \alpha}=[\alpha_0,\alpha_1,\cdots, \alpha_{Z_x-1}]^T$ and ${\bm \beta}=[\beta_0,\beta_1,\cdots, \beta_{Z_y-1}]^T$ for delay and Doppler shifts, respectively, such that 
\begin{align}
	\label{Weight:Condition:alpha:delay}
	\sum_{p=0}^{Z_x-1}\alpha_p=1, \ \ \alpha_p\geq 0, \ p=0,1,\cdots, Z_x-1, \\
	\label{Weight:Condition:beta:Doppler}
	\sum_{u=0}^{Z_y-1}\beta_u=1, \ \ \beta_u\geq 0, \ u=0,1,\cdots, Z_y-1.
\end{align}We define a matrix $\textbf{X}=[\textbf{X}^1,\textbf{X}^2,\cdots, \textbf{X}^M]^T$ of size $MZ_xZ_y \times (2(L_x+L_s)-1)$ from $\textbf{circ}\left(\textit{\textbf{x}}^{m,u}\right)$, where each matrix $\textbf{X}^m$ is given by
$\textbf{X}^m=\begin{bmatrix}
\textbf{x}(\textit{\textbf{x}}^{m,0}) ,
\textbf{x}(\textit{\textbf{x}}^{m,1}) ,
\cdots ,
\textbf{x}(\textit{\textbf{x}}^{m,Z_y-1}) 
\end{bmatrix}^T$, where each matrix $\textbf{x}(\textit{\textbf{x}}^{m,u})$ of size $Z_x\times (2(L_x+L_s)-1)$ is a weighted matrix constructed by taking the first $Z_x$ rows of $\textbf{circ}\left(\textit{\textbf{x}}^{m,u}\right)$, defined by
\begin{align}
	\textbf{x}(\textit{\textbf{x}}^{m,u})=\begin{bmatrix}
		\sqrt{\alpha_0}\sqrt{\beta_u}\textbf{circ}\left(\textit{\textbf{x}}^{m,u}\right)_0 \\
		\sqrt{\alpha_1}\sqrt{\beta_u}\textbf{circ}\left(\textit{\textbf{x}}^{m,u}\right)_1 \\
		\vdots \\
		\sqrt{\alpha_{Z_x-1}}\sqrt{\beta_u}\textbf{circ}\left(\textit{\textbf{x}}^{m,u}\right)_{Z_x-1}
	\end{bmatrix},
\end{align}where $\textbf{circ}\left(\textit{\textbf{x}}^{m,u}\right)_p$ represents the $p$-th row of the circulant matrix $\textbf{circ}\left(\textit{\textbf{x}}^{m,u}\right)$ for $p=0,1,\cdots, Z_x-1$, $u=0,1,\cdots, Z_y-1$, and $m=1,2,\cdots, M$. Similarly, we define a matrix $\textbf{Y}=[\textbf{Y}^1,\textbf{Y}^2,\cdots, \textbf{Y}^M]^T$ of size $MZ_xZ_y \times (2(L_x+L_s)-1)$ from $\textbf{circ}\left(\textit{\textbf{y}}^{n,v}\right)$ following the above process.

From the above two matrices $\textbf{X}$ and $\textbf{Y}$, we derive an upper bound and a lower bound on the Frobenius norm of $\textbf{X}\textbf{Y}^H$, denoted by $||\textbf{X}\textbf{Y}^H||_F$.

\begin{lemma}[Upper Bound of $||\textbf{X}\textbf{Y}^H||_F$]
	\label{Lemma:Upper:Bound:XYH}
	An upper bound of $||\textbf{X}\textbf{Y}^H||_F$ associated to the LAZ ${\bm \prod}$ with $1\leq Z_x\leq L_x$ and $1\leq Z_y \leq L_y$ is given by
	{\small 
	\begin{align}
		||\textbf{X}\textbf{Y}^H||_F^2\leq \sum_{m=1}^{M}\sum_{p=0}^{Z_x-1}\sum_{u=0}^{Z_y-1}\Big(\Big|{\bm \phi}_{{\max}_0}^{m,m}\Big|^2 -{\bm \phi}_{{\max}}^2 \Big) \alpha_p^2\beta_u^2+M^2{\bm \phi}_{{\max}}^2.
	\end{align}}
\end{lemma}
\begin{IEEEproof}
	We have 
{\small 
	\begin{align}
		&||\textbf{X}\textbf{Y}^H||_F^2
		=\sum_{m,n=1}^{M} \sum_{p,q=0}^{Z_x-1} \sum_{u,v=0}^{Z_y-1}\Big|{\bm \phi}^{m,n}(p-q,u-v)\Big|^2  \alpha_p\alpha_q\beta_u\beta_v \nonumber \\
		&\leq \sum_{m=1}^{M} \sum_{p=0}^{Z_x-1} \sum_{u=0}^{Z_y-1} \Big|{\bm \phi}_{{\max}_0}^{m,m}\Big|^2 \alpha_p^2\beta_u^2+\sum_{m=1}^{M} \sum_{p=0}^{Z_x-1}  \sum_{u,v=0 \atop u\neq v}^{Z_y-1}{\bm \phi}_{{\max}}^2\alpha_p^2\beta_u\beta_v \nonumber \\
		&+\sum_{m,n=1 \atop m\neq n}^{M} \sum_{p=0}^{Z_x-1} \sum_{u=0}^{Z_y-1}{\bm \phi}_{{\max}}^2\alpha_p^2\beta_u^2+ \sum_{m,n=1 \atop m\neq n}^{M} \sum_{p=0}^{Z_x-1}  \sum_{u,v=0 \atop u\neq v}^{Z_y-1}{\bm \phi}_{{\max}}^2\alpha_p^2\beta_u\beta_v \nonumber \\
		&+\sum_{m=1}^{M}  \sum_{p,q=0 \atop p\neq q}^{Z_x-1}\sum_{u=0}^{Z_y-1}{\bm \phi}_{{\max}}^2\alpha_p\alpha_q\beta_u^2 +\sum_{m=1}^{M}  \sum_{p,q=0, \atop p\neq q}^{Z_x-1} \sum_{u,v=0 \atop u\neq v}^{Z_y-1}{\bm \phi}_{{\max}}^2\alpha_p\alpha_q\beta_u\beta_v \nonumber \\
		&+\sum_{m,n=1 \atop m\neq n}^{M}  \sum_{p,q=0, \atop p\neq q}^{Z_x-1} {\bm \phi}_{{\max}}^2\alpha_p\alpha_q{\bm \beta}^T{\bm \beta}+ \sum_{m,n=1 \atop m\neq n}^{M} \sum_{p,q=0 \atop p\neq q}^{Z_x-1} \sum_{u,v=0 \atop u\neq v}^{Z_y-1}{\bm \phi}_{{\max}}^2\alpha_p\alpha_q\beta_u\beta_v \nonumber \\
		&= \sum_{m=1}^{M}\Big|{\bm \phi}_{{\max}_0}^{m,m}\Big|^2{\bm \alpha}^T{\bm \alpha}{\bm \beta}^T{\bm \beta}+M{\bm \phi}_{{\max}}^2{\bm \alpha}^T{\bm \alpha}(1-{\bm \beta}^T{\bm \beta}) \nonumber \\
		&+(M^2-M){\bm \phi}_{{\max}}^2{\bm \alpha}^T{\bm \alpha}{\bm \beta}^T{\bm \beta} +(M^2-M){\bm \phi}_{{\max}}^2{\bm \alpha}^T{\bm \alpha}(1-{\bm \beta}^T{\bm \beta}) \nonumber \\
		&+M{\bm \phi}_{{\max}}^2(1-{\bm \alpha}^T{\bm \alpha}){\bm \beta}^T{\bm \beta} +M{\bm \phi}_{{\max}}^2(1-{\bm \alpha}^T{\bm \alpha})(1-{\bm \beta}^T{\bm \beta}) \nonumber \\
		&+(M^2-M){\bm \phi}_{{\max}}^2(1-{\bm \alpha}^T{\bm \alpha}){\bm \beta}^T{\bm \beta} \nonumber \\
		&+(M^2-M){\bm \phi}_{{\max}}^2(1-{\bm \alpha}^T{\bm \alpha})(1-{\bm \beta}^T{\bm \beta}) \nonumber \\
		&=\sum_{m=1}^{M}\sum_{p=0}^{Z_x-1}\sum_{u=0}^{Z_y-1}\Big(\Big|{\bm \phi}_{{\max}_0}^{m,m}\Big|^2 -{\bm \phi}_{{\max}}^2 \Big) \alpha_p^2\beta_u^2+M^2{\bm \phi}_{{\max}}^2.
\end{align}}This completes the proof.
\end{IEEEproof}
We now derive the lower bound of $||\textbf{X}\textbf{Y}^H||_F$ in the following lemma.
\begin{lemma}[Lower Bound of $||\textbf{X}\textbf{Y}^H||_F$]
	\label{Lemma:Lower:Bound:XYH}
	A lower bound of $||\textbf{X}\textbf{Y}^H||_F$ associated to the LAZ ${\bm \prod}$ with $1\leq Z_x\leq L_x$ and $1\leq Z_y \leq L_y$ is given by
		\begin{align}
			||\textbf{X}\textbf{Y}^H||_F^2\geq \frac{\Big|\sum_{m=1}^{M}{\bm \phi}^{m,m}(0,0)\Big|^2}{l_{Z_x,Z_y}^{L_x,L_y}},
	\end{align}where $l_{Z_x,Z_y}^{L_x,L_y}=\min\{MZ_xZ_y, L_x+L_y-1\}$.
\end{lemma}
\begin{IEEEproof}
	A lower bound of $||\textbf{X}\textbf{Y}^H||_F$ can be expressed in terms of the rank and trace of the matrix $\textbf{X}\textbf{Y}^H$ as follows \cite{2021FuChen}:
	\begin{align}
		\label{lower:bound:norm:XYH}
		||\textbf{X}\textbf{Y}^H||_F^2\geq \Big|\text{tr}\Big(\textbf{X}\textbf{Y}^H\Big)\Big|^2/\text{rank}(\textbf{X}\textbf{Y}^H),
	\end{align}where $\text{tr}\Big(\textbf{X}\textbf{Y}^H\Big)$ is the trace of $\textbf{X}\textbf{Y}^H$ and $\text{rank}(\textbf{X}\textbf{Y}^H)$ is the rank of $\textbf{X}\textbf{Y}^H$. Based on the definitions of the matrices $\textbf{X}$ and $\textbf{Y}$, the rank of both the matrices $\textbf{X}$ and $\textbf{Y}$ is given by
	\begin{align}
		\text{rank}(\textbf{X})=\text{rank}(\textbf{Y})=\min\{MZ_xZ_y, L_x+L_y-1\}.
	\end{align}Therefore, the rank of $\textbf{X}\textbf{Y}^H$ is given by
	\begin{align}
		\label{rank:XYH}
		\text{rank}\Big(\textbf{X}\textbf{Y}^H\Big)\leq \min\{MZ_xZ_y, L_x+L_y-1\}.
	\end{align}The trace of $\textbf{X}\textbf{Y}^H$ is calculated as follows:
	\begin{align}
		\label{trace:XYH}
		\text{tr}\Big(\textbf{X}\textbf{Y}^H\Big)&=\sum_{m=1}^{M}\sum_{p=0}^{Z_x-1}\sum_{u=0}^{Z_y-1}{\bm \phi}^{m,m}(0,0)\alpha_p\beta_u \nonumber\\
		&=\sum_{m=1}^{M}{\bm \phi}^{m,m}(0,0).
	\end{align}Based on (\ref{rank:XYH}) and (\ref{trace:XYH}), from (\ref{lower:bound:norm:XYH}), we have 
	\begin{align}
		||\textbf{X}\textbf{Y}^H||_F^2\geq \frac{\left|\sum_{m=1}^{M}{\bm \phi}^{m,m}(0,0)\right|^2}{l_{Z_x,Z_y}^{L_x,L_y}},
	\end{align}where $l_{Z_x,Z_y}^{L_x,L_y}=\min\{MZ_xZ_y, L_x+L_y-1\}$. 
\end{IEEEproof}

According to \textit{Lemma \ref{Lemma:Upper:Bound:XYH}} and \textit{Lemma \ref{Lemma:Lower:Bound:XYH}}, we conclude a lower bound of the maximum AF magnitude ${\bm \phi}_{{\max}}$ associated to the LAZ ${\bm \prod}$ in the following theorem.
\begin{theorem}
	\label{Theorem:new:lower:bound}
	For any two weight vectors ${\bm \alpha }$ and ${\bm \beta}$ satisfying (\ref{Weight:Condition:alpha:delay}) and (\ref{Weight:Condition:beta:Doppler}) for delay and Doppler shifts, respectively, the lower bound of the maximum AF magnitude ${\bm \phi}_{{\max}}$ associated to the LAZ ${\bm \prod}$ with $1\leq Z_x\leq L_x$ and $1\leq Z_y \leq L_y$ for any length-$L_x$ unimodular sequence set and associated length-$L_y$ mismatched filter set is given by
	\begin{align}
		\label{proposed:lower:bound:AF}
		&{\bm \phi}_{{\max}}^2\geq \frac{1}{l_{Z_x,Z_y}^{L_x,L_y}\Big(M^2-M\sum_{p=0}^{Z_x-1}\sum_{u=0}^{Z_y-1} \alpha_p^2\beta_u^2\Big)}\times  \nonumber \\
		& \left(\left|\sum_{m=1}^{M}{\bm \phi}^{m,m}(0,0)\right|^2-l_{Z_x,Z_y}^{L_x,L_y}\sum_{m=1}^{M}\sum_{p=0}^{Z_x-1}\sum_{u=0}^{Z_y-1}\Big|{\bm \phi}_{{\max}_0}^{m,m}\Big|^2  \alpha_p^2\beta_u^2\right),
	\end{align}where $l_{Z_x,Z_y}^{L_x,L_y}=\min\{MZ_xZ_y, L_x+L_y-1\}$ and the maximum AF magnitude ${\bm \phi}_{{\max}_0}^{m,m}$ when the delay shift $\tau=0$ is given by
	\begin{align}
		\Big|{\bm \phi}_{{\max}_0}^{m,m}\Big|^2=\max_{0\leq |\nu|\leq Z_y-1}\left\{\big|{\bm \phi}^{m,m}(0,\nu)\big|^2\right\}.
	\end{align} 
\end{theorem}
	\subsection{Optimal Weight Vectors}
	From (\ref{proposed:lower:bound:AF}), we observe that the generalized lower bound is a function of two weight vectors ${\bm \alpha}$ and ${\bm \beta}$. The weight vectors ${\bm \alpha}$ and ${\bm \beta}$ are called optimal weight vectors if it offers the tighter bound. We also observe that the lower ${\bm \alpha}^T{\bm \alpha}=\sum_{p=0}^{Z_x-1}\alpha_p^2$ and ${\bm \beta}^T{\bm \beta}=\sum_{u=0}^{Z_y-1}\beta_u^2$, the larger the bound. Thus, optimizing the weight vectors ${\bm \alpha}$ and ${\bm \beta}$ such that ${\bm \alpha}^T{\bm \alpha}$ and ${\bm \beta}^T{\bm \beta}$ are minimized. By using Cauchy-Schwarz inequality, we have ${\bm \alpha}^T{\bm \alpha}=\sum_{p=0}^{Z_x-1}\alpha_p^2\geq \frac{\left(\sum_{p=0}^{Z_x-1}\alpha_p\right)^2}{Z_x}=\frac{1}{Z_x}$ and ${\bm \beta}^T{\bm \beta}=\sum_{u=0}^{Z_y-1}\beta_u^2\geq \frac{\left(\sum_{u=0}^{Z_y-1}\beta_u\right)^2}{Z_y}=\frac{1}{Z_y}$. The minimum values of ${\bm \alpha}^T{\bm \alpha}$ and ${\bm \beta}^T{\bm \beta}$ are obtained if and only if ${\bm \alpha}=\left[\frac{1}{Z_x}, \frac{1}{Z_x},\cdots, \frac{1}{Z_x}\right]^T$ and ${\bm \beta}=\left[\frac{1}{Z_y}, \frac{1}{Z_y}, \cdots, \frac{1}{Z_y}\right]^T.$
	
	\begin{corollary}[Maximum AF Magnitude Lower Bound for Optimal Weight Vectors ${\bm \alpha}$ and ${\bm \beta}$]
	For the optimal weight vectors ${\bm \alpha }$ and ${\bm \beta}$, the lower bound of the maximum AF magnitude ${\bm \phi}_{{\max}}$ associated to the LAZ ${\bm \prod}$ with $1\leq Z_x\leq L_x$ and $1\leq Z_y \leq L_y$ for any length-$L_x$ unimodular sequence set and associated length-$L_y$ mismatched filter set is given by
		  \begin{align}
		  	\label{lower:bound:with:optimal:weights}
		  	&{\bm \phi}_{{\max}}^2\geq \frac{1}{l_{Z_x,Z_y}^{L_x,L_y}\Big(M^2Z_xZ_y-M\Big)}\times  \nonumber \\
		  	& \left(Z_xZ_y\left|\sum_{m=1}^{M}{\bm \phi}^{m,m}(0,0)\right|^2-l_{Z_x,Z_y}^{L_x,L_y}\sum_{m=1}^{M}\Big|{\bm \phi}_{{\max}_0}^{m,m}\Big|^2  \right).
		  \end{align}
	\end{corollary}Since ${\bm \phi}^{m,m}(0,0)$ and ${\bm \phi}_{{\max}_0}^{m,m}$ can be expressed through the energy of the filter (or Doppler shifted filter), we can set ${\bm \phi}^{m,m}(0,0)={\bm \phi}_{{\max}_0}^{m,m}=1$.
	
	Then, (\ref{lower:bound:with:optimal:weights}) becomes
	\begin{align}
		\label{Optimal:weight:Lower:Bound}
		{\bm \phi}_{{\max}}^2 \geq \frac{MZ_xZ_y-l_{Z_x,Z_y}^{L_x,L_y} }{l_{Z_x,Z_y}^{L_x,L_y}\Big(MZ_xZ_y-1\Big)}.
	\end{align}

	\section{Relationship With Existing Bounds}
	\label{Sec:Rela:Existing:Litr}
	
	\subsection{Relationship With the Well-Known Welch Bound \cite{1974Welch}}
	We show that the proposed lower bound includes Welch bound \cite{1974Welch} as a special case. Let us consider a special case of (\ref{Optimal:weight:Lower:Bound}) defined by $Z_x=2L_x-1$ and $Z_y=1$ with constant weight vectors ${\bm \alpha}=[1/Z_x,1/Z_x,\cdots, 1/Z_x]^T$ and ${\bm \beta}=[1/Z_y]^T$ with $Z_y=1$. Then, the Doppler shift becomes $\nu=0$. We have $l_{Z_x,Z_y}^{L_x,L_y}=2L_x-1$. Then, (\ref{Optimal:weight:Lower:Bound}) becomes  
	\begin{align}
		\label{Welch:Lower:Bound:Special:Case}
		{\bm \phi}_{{\max}}^2 \geq \frac{M-1 }{M(2L_x-1)-1},
	\end{align}which is the well-known Welch bound \cite{1974Welch}.

	\subsection{Relationship With Tan-Arlery-Rabaste-Lehmann-Ovarlez Correlation Lower Bound \cite{2020Tan-Arlery-Rabaste-Lehmann-OvarlezTIT}}
	We also show that the proposed lower bound (\ref{proposed:lower:bound:AF}) includes Tan-Arlery-Rabaste-Lehmann-Ovarlez correlation lower bound \cite{2020Tan-Arlery-Rabaste-Lehmann-OvarlezTIT} as a special case. Let us consider a special case of the proposed lower bound (\ref{proposed:lower:bound:AF}) defined by $Z_x=L_x=K$ and $Z_y=1$. We consider a weight vector ${\bm \alpha}=[\alpha_0,\alpha_1,\cdots, \alpha_{K-1}]^T$ satisfying (\ref{Weight:Condition:alpha:delay}) and ${\bm \beta}=[1/Z_y]^T$ with $Z_y=1$. The Doppler shift becomes $\nu=0$. Then, AF ${\bm \phi}^{m,m}(\tau,0)$ is the correlation function ${\bm \theta}^{m,m}(\tau)$ for each $m=1,2,\cdots M$. Consequently, (\ref{proposed:lower:bound:AF}) becomes 
	\begin{align}
		\label{Tan-Arlery-Rabaste-Lehmann-Ovarlez:bound}
		&{\bm \phi}_{{\max}}^2\geq \frac{1}{\Big(M^2-M{\bm \alpha}^T{\bm \alpha}\Big)}  \nonumber \\
		& \times \left(\frac{\left|\sum_{m=1}^{M}{\bm \phi}^{m,m}(0,0)\right|^2}{l_{K}^{L_x,L_y}}-{\bm \alpha}^T{\bm \alpha}\sum_{m=1}^{M}\Big|{\bm \phi}_{{\max}_0}^{m,m}\Big|^2  \right),
	\end{align}where $l_{K}^{L_x,L_y}=\min\{MK, L_x+L_y-1\}$, which is the Tan-Arlery-Rabaste-Lehmann-Ovarlez lower bound \cite{2020Tan-Arlery-Rabaste-Lehmann-OvarlezTIT}. 
	
	\section{Conclusion}
	\label{Sec:conclusion}
	In this work, we investigated a new lower bound on the maximum AF sidelobes over LAZ  ${\bm \prod}$ in the delay-Doppler domain as demonstrated in \textit{Theorem \ref{Theorem:new:lower:bound}}. The main innovation is to consider a set of mismatched filters associated to a unimodular sequence set and two weight vectors corresponding to the delay and Doppler shifts, respectively. As an extension of the previous Tan-Arlery-Rabaste-Lehmann-Ovarlez lower bound \cite{2020Tan-Arlery-Rabaste-Lehmann-OvarlezTIT}, the proposed AF lower bound associated to the LAZ  ${\bm \prod}$ is a function of the set size $M$, the sequence length $L_x$, the filter length $L_y$, delay shift $Z_x$, Doppler shift $Z_y$, the weight vectors ${\bm \alpha}, {\bm \beta}$ and the AF values ${\bm \phi}_{{\max}_0}^{m,m}$ and ${\bm \phi}^{m,m}(0,0)$ for all $m=0,1,\cdots, M$. Furthermore, we show a relationship between the proposed lower bound and the existing Welch bound \cite{1974Welch}. As a future work, it would be interesting to design a set of unimodular sequences and its associated mismatched filter set meeting the proposed AF lower bound. 
		

\end{document}